\documentclass[conference]{IEEEtran}
\IEEEoverridecommandlockouts
\usepackage{cite}
\usepackage{amsmath,amssymb,amsfonts,amsthm,mathtools}
\usepackage{algorithmic}
\usepackage{graphicx}
\usepackage{textcomp}
\usepackage{xcolor}
\usepackage{tikz}
\usepackage{tikz-cd}

\newcommand{\GG}{\mathcal{G}}
\newcommand{\VV}{\mathcal{V}}
\newcommand{\EE}{\mathcal{E}}
\newcommand{\TT}{\mathcal{T}}
\newcommand{\XX}{\mathcal{X}}
\newcommand{\YY}{\mathcal{Y}}
\newcommand{\bin}{\text{Bin}}
\newcommand{\ceil}[1]{\lceil #1 \rceil}

\DeclareMathOperator{\rank}{rank}

\newtheoremstyle{ieeestyle}
  {3pt}
  {3pt}
  {\normalfont}
  {}
  {\bfseries}
  {.}
  { }
  {}

\theoremstyle{ieeestyle}

\newtheorem{thm}{Theorem}
\newtheorem{lem}{Lemma}
\newtheorem{prop}{Proposition}
\newtheorem{defn}{Definition}
\newtheorem{rmk}{Remark}
\newtheorem{exam}{Example}
\newtheorem{colry}{Corollary}

\def\BibTeX{{\rm B\kern-.05em{\sc i\kern-.025em b}\kern-.08em
    T\kern-.1667em\lower.7ex\hbox{E}\kern-.125emX}}

\begin{document}

\title{A discrete Benamou-Brenier formulation of Optimal Transport on graphs\\

\thanks{This research was supported by the EPSRC Centre for Doctoral Training in Computational Statistics and Data Science (COMPASS) – funded by the Engineering and Physical Sciences Research Council [grant number EP/S023569/1].}
}

\author{\IEEEauthorblockN{ Kieran Morris }
\IEEEauthorblockA{\textit{School of Mathematics} \\
\textit{University of Bristol}\\
Bristol, UK \\
k.morris@bristol.ac.uk}
\and
\IEEEauthorblockN{ Oliver Johnson}
\IEEEauthorblockA{\textit{School of Mathematics} \\
\textit{University of Bristol}\\
Bristol, UK \\
O.Johnson@bristol.ac.uk}

}

\maketitle

\begin{abstract}
We propose a discrete transport equation on graphs. We then derive a discrete analogue of the Benamou-Brenier formulation for Wasserstein-$1$ distance on a graph and as a result classify all constant speed $W_1$ geodesics on graphs.
\end{abstract}


\section{Introduction}
The classical transport problem considers two measures $\mu$ and $\nu$ on spaces $\XX$ and $\YY$, where we have cost (and in our case a metric) of transporting $x \in \XX$ to $y \in \YY$, which we denote $d(x,y)$. Then we ask: how can we transport $\mu$ to $\nu$ while minimising the expected cost? This gives us a measure of discrepancy between $\mu$ and $\nu$ relative to $d(x,y)$. The Kantorovich formulation of Optimal Transport \cite{kantorovitch1958translocation} considers the minimum transport cost over couplings of $\mu$ and $\nu$.
\begin{defn}[Kantorovich Formulation of optimal transport]\label{defn: kant}
    Let $\mu,\nu$ be a measures on $\XX$ and $d$ a metric on $\XX$, we define the set of couplings of $\mu$ and $\nu$ to be $
    \Pi(\mu,\nu) := \{ \pi \in \mathcal{P}(\XX^2) : \pi(A \times \XX) = \mu(A), \pi(\XX \times B) = \nu(B) \}$, and the Wasserstein $p$ distance by:
    \begin{equation}\label{eq: kant}
        W_p(\mu,\nu)^p = \inf_{\pi \in \Pi(\mu,\nu)} \left(\int_{\XX^2} d(x,y)^p d\pi\right).
    \end{equation}
\end{defn}
It is known that $W_p$ is a metric on the space of (finite $p$'th moment) measures on $\XX$, see \cite[Definition 6.4]{villani2008optimal}. The Wasserstein distance is popular in machine learning circles, namely as a loss function \cite{arjovsky2017wasserstein,frogner2015learning,ozair2019wasserstein} and utilising its Riemannian structure for Wasserstein Gradient Flows \cite{jordan1998variational,mokrov2021large}. One notable development in optimal transport came via Benamou and Brenier \cite{benamou2000computational}, who re-parametrised the minimisation in terms of time-dependent distribution and velocity field pairs $(f_t,v_t)$.
\begin{prop}[Benamou-Brenier Formulation]\label{prop: benamou}
    Given two distributions $f(0),f(1) \in \mathcal{P}(\mathbb{R}^d)$, we can then express
    \begin{equation}\label{eq: benamou}
        W_p(f(0),f(1)) = \inf_{(f_t,v_t)} \left(\int_0^1 \int ||v_t||^p df_t dt \right)^{\frac{1}{p}}
    \end{equation}
     where $(f_t,v_t)$ are constrained by the transport equation:
    \begin{equation}\label{eq: continuity}
        \partial_t f_t + \nabla \cdot (f_t v_t) = 0,
    \end{equation}
    where $\nabla \cdot$ is the divergence operator, (see \cite{benamou2000computational}).
\end{prop}
\begin{defn}[Constant Speed Geodesics]\label{defn: constant}
Given a metric space $(\XX,d)$, we say a path $f_t$ is a \textit{Constant Speed Geodesic} if and only if $d(f_s,f_t) = |s-t|d(f_0,f_1)$ and for a general path we define its \textit{Speed} as $|\dot{f}|_d = \lim_{h \rightarrow 0} \frac{1}{h}d(f_{t +h},f_t)$. Constant speed geodesics have constant speed, see \cite[Remark 7.7]{villani2008optimal}.
\end{defn}
The infimum in \eqref{eq: benamou} can be achieved by a constant speed $W_p$ geodesic $f_t$, and $v_t$ satisfying $||v_t|| = |\dot{f}|_{W_p}$, (see \cite{benamou2000computational}).\newline \newline
On $\mathbb{R}^d$, the constant speed $W_p$ geodesics are the McCann displacement interpolations, see \cite[Definition 1.1]{mccann1997convexity}. For discrete measures, say on $\{0,1...,n \} \subset \mathbb{R}$, the McCann interpolations will traverse $\mathbb{R}$, but not remain in the domain $\{0,1...,n \}$. This means to interpolate across discrete domains, considering them as embedded in a continuum will not suffice.
\newline \newline
On a graph $\GG = (\VV,\EE)$, computing \eqref{eq: kant} is a $|\VV|^2$-dimensional linear program. However, Beckmann \cite{beckmann1952continuous} gave a time-independent velocity field formulation of $W_1$ which on a graph is $|\EE|$-dimensional, (see \cite[Section 6.3]{peyre2019computational}). In fact the Benamou-Brenier formulation can be viewed as a generalisation of Beckmann's work which introduces time parameters.
\newline \newline
The main contribution to a discrete Benamou-Brenier analogue came from Erbar and Maas \cite{maas2011gradient}, who formulated an analogue of $W_2$ on discrete domains by considering a transport equation on graphs and a kinetic energy functional much like in Proposition \ref{prop: benamou}. This metric is Riemannian and interplays well  with discrete gradient flows, although importantly it is \textit{not} the $W_2$ distance. This left the problem of a Benamou-Brenier interpretation of $W_p$ on graphs open.
\newline \newline
Later, Hillion and Johnson \cite{hillion2016} provided sufficient conditions for a true discrete Benamou-Brenier form of $W_1$ on $\mathbb{Z}$, which is the work we are building on.
In this paper we offer our own discrete transport equation involving triples $(f,v,g)$, and derive a Benamou-Brenier formulation for $W_1$ on trees  (Lemma \ref{lem: const}, Theorem \ref{thm: benamoutree}). We then generalise this result to graphs (Lemma \ref{lem: const2}, Theorem \ref{thm: benamou}) and classify the constant speed $W_1$ geodesics via the solutions to our formulation (Proposition \ref{thm: geod}).

\section{Notation and the Discrete Transport Equation}
\begin{exam}\label{examp: bin1}
Consider the path of Binomial distributions $f(t) = \bin(n,p(t))$ on the set $\VV = \{0,1,...,n\}$ where $p(t)$ varies with time. We show the case $n=5$ below:

\begin{center}
\begin{tikzpicture}
  \def\n{5}          
  \def\spacing{1.4}  

  \foreach \i in {0,...,\n} {
    \coordinate (P\i) at ({\i*\spacing}, 0);
    \filldraw[black] (P\i) circle (2pt);
    \node[above] at (P\i) {$f_{\i}$};
  }

  \foreach \i in {0,...,\numexpr\n-1\relax} {
    \pgfmathtruncatemacro{\j}{\i+1}
    \draw (P\i) -- (P\j);
  }
\end{tikzpicture}

\end{center}
By the product rule and some rudimentary combinatorics we can express the derivative as:
\begin{align*}
    \partial_t f_x &= np'\bin_{x-1}(n-1,p) - np'\bin_{x}(n-1,p)\\
    &= -\nabla_1 (np'\bin_x(n-1,p))
\end{align*}
where $\nabla_1(h_x) := h_{x} - h_{x-1}$ is the finite difference operator.
\end{exam}
This resembles \eqref{eq: continuity}, yet the distribution inside the spatial derivative is no longer $f$, in fact it is another distribution $g$, that has support $\{0,1,...,n-1\}$. This means the discrete transport equation satisfied in this case is 
\begin{equation*}
    \partial_t f + \nabla_1 (vg) = 0
\end{equation*}
where $v(t) = np'$ and $g(t) = \bin (n-1,p)$. Our key insight is that $\VV$ can be considered as the vertices of a path graph $\GG$, and that $g$ and $v$ exist on the edges of $\GG$ as follows:
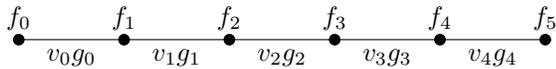
\begin{center}

\begin{tikzpicture}
  \def\n{5}          
  \def\spacing{1.4}  

  \foreach \i in {0,...,\n} {
    \coordinate (P\i) at ({\i*\spacing}, 0);
    \filldraw[black] (P\i) circle (2pt);
    \node[above] at (P\i) {$f_{\i}$};
  }

  \foreach \i in {0,...,\numexpr\n-1\relax} {
    \pgfmathtruncatemacro{\j}{\i+1}
    \draw (P\i) -- node[below] {$v_{\i}g_{\i}$} (P\j);
  }
\end{tikzpicture}

\end{center}
and the discrete transport equation states that the rate of change at node $x$ is the potential difference across that node. On a general graph, we do not have the canonical rightward direction of edges for mass to flow, so in some sense on a graph $\GG = (\VV,\EE)$ we must choose this orientation arbitrarily.
\subsection{Matrices, Operators, Equations and Integrals}
In this section, we will define our necessary objects. This includes a divergence operator, a discrete transport equation, and an energy functional, which we aim to minimise, as constrained by the discrete transport equation.
\begin{defn}[Incidence Matrix and Arrow Shorthand]
    Given a directed graph $\GG = (\VV,\EE)$, we define the incidence matrix  $\Omega = (\omega_{x,k})_{x \in \VV,k\in \EE}$, by
    \begin{equation*}
    \omega_{x,k} =   \begin{cases}
             1 &: k \text{ is incoming to } x \\
             -1 &: k \text{ is outgoing from } x \\
             0 & \text{otherwise.}
        \end{cases}
    \end{equation*}
    As shorthand for $\omega_{x,k} = 1$ and $\omega_{k,x} = -1$ we will write $k \rightarrow x$ and $x \rightarrow k$ respectively.
\end{defn}
For example, we could orient $\mathbb{Z}$ by directing edges away from $0$, rather than left to right.

\begin{defn}[Floor and Ceiling Notation]
    Given an edge $k \in \EE$, we denote $\lceil k \rceil$ to be the node it is incoming to, and $\lfloor k \rfloor$ to be the node it is outgoing from.
\end{defn}
Now, we can use the incidence matrix to define discrete derivative and divergence operators between functions on vertices and edges respectively.
\begin{defn}[The Gradient and Divergence Operator] \mbox{ }
Let $\GG = (\VV,\EE)$ be a graph with incidence matrix $\Omega$. Then 
\begin{enumerate}
\item For a function $f$ defined on vertices $\VV$, we define the \textit{Gradient Operator} at edge $k \in \EE$ to be: 
$$ (\nabla f)_k := (-\Omega^T \cdot f)_k =  -\sum_{x \in \VV} \omega_{x,k} f_x =f_{\lfloor k \rfloor} - f_{\lceil k \rceil},$$
i.e., the difference in function $f$ along the edge $k$.
\item
    For a function $g$ defined on the edges $\EE$, we define the \textit{Divergence Operator} at vertex $x \in \VV$ to be:
    \begin{equation*}
(\nabla \cdot g)_x := -\Omega \cdot g = -\sum_{k \in \EE} \omega_{x,k} g_k
= \sum_{k:x \rightarrow k} g_k - \sum_{k:k \rightarrow x} g_k,\end{equation*} 
i.e., the difference in total flow out of and  flow into $x$.
\end{enumerate}
\end{defn}
If we consider $\nabla_1$ as instead defined on edges $(x,x-1)$ of $\mathbb{Z}$, it corresponds exactly to the $\nabla \cdot$ operator.
\subsection{The Discrete Transport Equation}
We now build a framework to generalise Example \ref{examp: bin1}. We start by building the transport equation as satisfied by a triple $(f,v,g)$, where we impose that $v$ and $g$ exist on edges.
\begin{defn}[Discrete Transport Equation]
    Given a graph $\GG = (\VV,\EE)$ with incidence matrix $\Omega$, we say the transport equation is satisfied if: for a path of distributions on vertices $f = f(t)_x$, a path of functions on edges $v = v(t)_k$ (called the velocity), and a path of distributions on edges $g = g(t)_k$, we have that for each $x \in \VV$
    \begin{equation}\label{eq: original}
        \partial_t f(t)_x + (\nabla \cdot v(t)g(t))_x \text{ for all }  t \in [0,1].
    \end{equation}
    This can be compactly expressed in terms of the incidence matrix $\Omega$:
    \begin{equation}\label{eq: matrix}
        \partial_t f = \Omega \cdot (vg)
    \end{equation}
    where the product $vg$ is component-wise.
\end{defn}
\begin{exam}\label{exam: stargraph1}
    Consider the following directed star graph $S_3$;
    
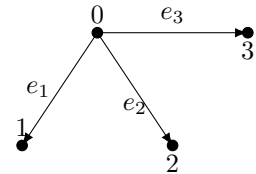
\begin{center}
\begin{tikzpicture}[>=Latex]
    \coordinate (A0) at (0,0);
    \coordinate (A1) at (-1,-1.5);
    \coordinate (A2) at (1,-1.5);
    \coordinate (A3) at (2,0);

    \filldraw[black] (A0) circle (2pt);
    \filldraw[black] (A1) circle (2pt);
    \filldraw[black] (A2) circle (2pt);
    \filldraw[black] (A3) circle (2pt);

    \node[above] at (A0) {$0$};
    \node[above] at (A1) {$1$};
    \node[below] at (A2) {$2$};
    \node[below] at (A3) {$3$};
    \draw[->,>=Latex] (A0) -- (A1) node[midway,left] {$e_{1}$};
    \draw[->,>=Latex] (A0) -- (A2) node[midway,below] {$e_{2}$};
    \draw[->,>=Latex] (A0) -- (A3) node[midway,above] {$e_{3}$};
\end{tikzpicture}
\end{center}    
    and define $f$ to be $
        f(t) = (Z,Zs,Zs,Zs)$ for an arbitrary positive $s(t)$ and partition function $Z(t) = \frac{1}{1+3s}$. This is the stationary distribution of a Markov jump process $Q$ where $q_{x0}s(t) = q_{0x}$. Equation \eqref{eq: original} states that 
    \begin{align*}
        \partial_tf_0 & = -Z'(t) &&= -v_{e_1}g_{e_1} - v_{e_2}g_{e_2} - v_{e_3}g_{e_3},\\
        \partial_t f_1 &= \frac{1}{3}Z'(t) &&= v_{e_1}g_{e_1},\\
        \partial_t f_2 &= \frac{1}{3}Z'(t)  &&= v_{e_2}g_{e_2},\\
        \partial_t f_3 &= \frac{1}{3}Z'(t) &&= v_{e_3}g_{e_3}.
    \end{align*}
    \addtolength{\topskip}{2pt}
    From here we have many choices for $(v,g)$. One of note is when $v$ is invariant across edges, i.e., when we choose 
    \begin{equation*}
        v_{e_1} = v_{e_2} = v_{e_3} = Z'(t), \quad g_{e_1} = g_{e_2} = g_{e_3} = \frac{1}{3}.
    \end{equation*}
    We will return to this distribution and will motivate our choice of solution in Example \ref{examp: star graph}.
\end{exam}

We have our discrete transport equation, but a Benamou-Brenier formulation also requires an energy functional to minimise. In Proposition \ref{prop: benamou}, we compute the expectation of $v$ with respect to $f$. However since $v$ exists on the edges with $g$, it is more natural to take the expectation over $g$. We keep the integral over time unchanged.
\begin{defn}[Integral Formulation]
    Let $\GG = (\VV,\EE)$ be a graph with incidence matrix $\Omega$ and suppose $(f,v,g)$ satisfy \eqref{eq: original}. For $q \geq 1$ we define
    \begin{equation}\label{eq: I}
        \mathcal{I}_q(v,g) \coloneqq \left(\int_0^1 \sum_{k \in \EE} g_k(t) |v_k(t)|^q dt\right)^{\frac{1}{q}}.
    \end{equation}
    Then for two distributions $f(0),f(1)$ on $\VV$, we define 
    \begin{equation}\label{eq: vg}
        V_q(f(0),f(1)) \coloneqq \inf_{f,v,g} \left\{ \mathcal{I}_q(v,g)  : \partial_t f = \Omega \cdot (vg) \right\}
    \end{equation}
    and so $V_q$ is a measure of discrepancy between $f(0)$ and $f(1)$. This is a generalisation of a functional in \cite[Definition 3.2]{hillion2016}, which is the work we are building on.
\end{defn}
\section{Main Results}
\subsection{Benamou-Brenier on the Vertices of a Tree}\label{sect: tree}
The Wasserstein-$1$ distance on $\mathbb{Z}$ has a closed form expression, see \cite[Lemma 8.2]{bickel1981some}. Slightly lesser known is that a generalised expression for the tree exists too -- thanks to Evans, \cite{evans2012phylogenetic}  -- which remarkably corresponds to a metric from Microbiology called the UniFrac metric, see \cite{lozupone2005unifrac}. Evans defines $W_1$ in terms of cuts, but for our purposes, we will formulate it in terms of tails. Let $\GG$ be a tree, and choose an arbitrary root node $r \in \VV$ and define our incidence matrix by directing arrows away from $r$. This allows us to define tails.
\begin{defn}[Tails of Distributions]
We define the tail  past $x \in \VV$ to be
\begin{equation*}
    U(x) = \{y \in \VV: \exists k_1,...k_n (x \rightarrow k_1 \rightarrow \cdots k_n\rightarrow y) \}
\end{equation*}
and then $F_x = \sum_{y \in \VV}\mathbb{I}(y \in U(x)) f_y$ is called the \textit{tail distribution of} $f$ \textit{past }$x$.
\end{defn}
$U(x)$ and $F_x$ generalise the set $\{x,x+1,x+2,...\}$ for $\mathbb{Z}$ and the tail distribution on $\mathbb{Z}$ respectively. Then we have the following result.
\begin{prop}[Wasserstein-$1$ on a Tree \cite{evans2012phylogenetic}]
    Given a rooted tree $\GG$, for two distributions $f(0)$ and $f(1)$ on $\VV$ we have that
    \begin{equation}\label{eq: tree}
    W_1(f(0),f(1)) = \sum_{x \in \VV}|F(1)_x - F(0)_x|.
\end{equation}
It is worth noting that we have other special cases of $W_1$. For cyclic graphs, we have a another closed form expression, (see \cite{rabin2011transportation}), which utilises cuts much like Evans \cite{evans2012phylogenetic}. For a general graph, $W_1$ can be expressed as the minimum transport cost across possible spanning trees of $\GG$, see \cite{guy2019kantorovich}. 
\end{prop}
\begin{lem}[Tail Variant of the Transport Equation]\label{thm: tail}
    Let $\GG$ be a tree with root node $r$ and incidence matrix $\Omega$, we can invert \eqref{eq: original} via the tail distributions:
\begin{equation} \partial_t F(t)_{\ceil{k}} =  v(t)_{k}g(t)_{k} \text{ for each }k \in \EE,\label{eq:tailpde} \end{equation}
and $g$ can be derived exactly from $f$ and $v$.
\end{lem}
\begin{proof}
    We drop the dependence on $t$ for now, and consider:
    \begin{align*}
        \partial_t F_{\ceil{k}} &=\sum_{y \in \VV} \mathbb{I}(y \in U(\ceil{k})) \partial_t f_y\\ &= \sum_{y \in \VV} \mathbb{I}(y \in U(\ceil{k}))\sum_{j \in \EE} \omega_{y,j}v_jg_j\\
        &= \sum_{j \in \EE}v_j g_j \sum_{y \in \VV} \omega_{y,j}\mathbb{I}(y \in U(\ceil{k})) \\
        &= \sum_{j \in \EE}v_jg_j (\nabla \mathbb{I}(y \in U(\ceil{k}))_j\\
        &= \sum_{j \in \EE} v_j g_j \mathbb{I}(j = k)\\
        &= v_kg_k.
    \end{align*}
    Then $g$ can be derived from $f$ and $v$ by rearranging \eqref{eq:tailpde}.
\end{proof}
While we have inverted \eqref{eq: original}, we still have a family of possible $(v,g)$ which satisfy \eqref{eq:tailpde}, below we discuss the most important case for our purposes.
\begin{lem}[Constant Speed Solutions]\label{lem: const}
    Given a path $f$, we define the \textit{Constant Speed Solution} to be
    \begin{equation}
        v_k = \text{sign}(\partial_t F_{\ceil{k}})|v|, \quad g_k = |\partial_t F_{\ceil{k}}|/|v|
    \end{equation}
     where $|v| = \sum_{x \in \VV} |\partial_t F(t)_x|$. For this we have $|v| = |\dot{f}|_{W_1}$, and if $f$ is a constant speed geodesic, we additionally have that $\mathcal{I}_2(v,g) = |v|= W_1(f(0),f(1))$.
\end{lem}
\begin{proof}
    See Appendix \ref{appendix: lem: const}.
\end{proof}

\begin{thm}[Benamou-Brenier across a tree]\label{thm: benamoutree}
    Let $\GG$ be a rooted tree, then we have
    \begin{equation}\label{eq: tree ben}
W_1(f(0),f(1)) = \inf_{f,g,v}\left\{\mathcal{I}_2(v,g) : \partial_t f = \Omega \cdot(vg)) \right\}.
\end{equation}
for all $f(0)$ and $f(1)$ on $\VV$. The minimum is achieved by a constant speed geodesic $f$, and the velocity $v$ satisfies $|v(t)_k| = |\dot{f}|_{W_1}$, analogous to the continuous case.
\end{thm}
\begin{proof}
    An immediate consequence of \eqref{eq: tree} and \eqref{eq:tailpde} is we are given an integral formulation of $W_1$. Remarkably, if $(f,v,g)$ is \textit{any} triple which satisfies \eqref{eq: original}, then we can write
    \begin{align*}
        W_1(f(0),f(1) &= \sum_{x \in \VV} | F(1)_x - F(0)_x| \\
        &= \sum_{x \in \VV} \left| \int_0^1 \partial_t F(t)_x dt \right| \\
        &= \sum_{k \in \EE} \left| \int_0^1 g(t)_kv(t)_k dt \right|
    \end{align*}
    which transforms a sum over tails to a sum over edges. Squaring both sides and applying both Jensen's and the Cauchy--Schwarz inequalities bounds $W_1$ above by $V_2$.
    \begin{align*}
        W_1(f(0),f(1))^2 &= \left( \sum_{k \in \EE} \left| \int_0^1 g(t)_kv(t)_kdt \right|\right)^2 \\
        &\leq \left( \int_0^1 \sum_{k \in \EE} |g(t)_kv(t)_k| dt \right)^2\\
        &\leq \int_0^1 \left( \sum_{k \in \EE} |g(t)_kv(t)_k|\right)^2 dt \\
        &\leq \int_0^1 \sum_{k \in \EE} g(t)_k|v(t)_k|^2 dt  = \mathcal{I}_2(v,g)^2
    \end{align*}
    So, we have that $V_2 \geq W_1$ too. However, Lemma \ref{lem: const} says that a constant speed solution achieves $W_1$, so choosing an arbitrary constant speed geodesic as in Definition \ref{defn: constant} - see Appendix \ref{appendix: geodesic} for an explicit example - we have that $V_2(f(0),f(1)) = W_1(f(0),f(1))$ and we are done.
\end{proof}
In fact, the proofs for Lemma \ref{lem: const} and Theorem \ref{thm: benamoutree} work for $q \geq 1$ too, so we have a more general result that $V_q = W_1$, see Appendix \ref{appendix: thm: benamoutree} , and we can choose $q$ arbitrarily.
\begin{exam}\label{examp: bin}
    Consider the Binomial case (Example \ref{examp: bin1}) when $p(t) = (1-t)p + tq$ for $p > q > 0$. We know that $v(t) =  n(q-p)$ and $g(t) = \bin(n-1,p(t))$ satisfy \eqref{eq: original}. We also have that $\mathcal{I}_2(v,g) = n(q-p) = W_1(f(0),f(1))$,
    so the original choice of $(f,v,g)$ is minimizing. This solution also contains a constant speed geodesic in the form of $f$.
\end{exam}

\begin{exam}\label{examp: poi}
    Let $f(t) = \text{Poi}((1-t)\lambda + t\mu)$, then $v(t) = \mu - \lambda$ and $g(t) = \text{Poi}((1-t)\lambda + t\mu)$ also  satisfy \eqref{eq: original} and $\mathcal{I}_2(v,g) = \lambda - \mu = W_1(f(0),f(1))$. Again $f$ is a constant speed geodesic.
\end{exam}
\begin{exam}\label{examp: star graph}
Consider Example \ref{exam: stargraph1}. When $s(t) = \frac{-1}{3}(1 + \frac{1}{at+b})$ for $a,b < -1$, we have that $s$ is both decreasing and positive. We can also show that $Z''(t) = 0$. So by choosing the edge invariant solution we know $Z'(t) = v_{01} = v_{02} = v_{03}$ are also constant in time. Plugging this into $\mathcal{I}_2$ we have
\begin{equation*}
    \mathcal{I}_2(v,g) = \int_0^1 Z'(t) dt = Z(1) - Z(0) = W_1(f(0),f(1)).
\end{equation*}
For this choice of $s(t)$, we have that $f(t) = (1-t)f(0) + tf(1)$ and by Appendix \ref{appendix: geodesic} this is a constant speed geodesic.
\end{exam}
We cannot apply the same argument to the graph case as in the tree, to work toward the graph, we will build a \textit{Reduced Formulation} of $V_q$ by studying the solutions to \eqref{eq: original} in more detail.
\subsection{The Reduced Formulation on a Graph}
In the continuous formulation, so long as smoothness assumptions were made, switching from $f_t$ to $v_t$ and vice versa was a matter of solving a first-order PDE, see  \cite[Chapter 16.1]{ambrosio2021lectures} for details. However, in our case, we can solve this more directly - although different methods are required for each direction.
\begin{rmk}\label{rmk: sol2}
Given $f$, solving \eqref{eq: matrix} for $(v,g)$ involves inverting $\Omega$. Fortunately there is a plethora of work on solving these systems, see \cite[Section 2]{bapat2010graphs}. For an incidence matrix $\Omega$, we have $\text{Im}(\Omega) = \{ x: \sum_j x_j = 0 \}$ and $\rank (\Omega) = |\VV| -1$, see \cite[Lemma 2.4]{bapat2010graphs}. We know $\partial_t f \in \text{Im} (\Omega)$ for all $t$, so we can always solve for the vector $vg$. Although this system is not full rank, removing an arbitrary row $r$ from $f$ and $\Omega$ gives
\begin{equation*}
    \partial_t \tilde{f} = \tilde{\Omega} \cdot (vg)
\end{equation*}
which is full rank. Then for an arbitrary right inverse $P$ (such that $ \tilde{\Omega}  P= I$) we can express the general solution as
\begin{equation*}
    vg = P \cdot \partial_t \tilde{f} + \epsilon
\end{equation*}
for some $\epsilon \in \ker(\Omega)$. One can find a right inverse by choosing a spanning tree $\TT$ and considering its \textit{path matrix} $P_{\TT}$, the right inverse then has block structure including $P_{\TT}$. This additionally characterises $\ker(\Omega)$ as the cycle space of $\Omega$, see \cite[Theorem 2.13]{bapat2010graphs}. In the case that $\GG$ is a tree, $\tilde{\Omega}$ is invertible, and $\tilde{\Omega}^{-1}$ is defined by \eqref{eq:tailpde}. For each valid $(P,\epsilon)$ we have a family of solutions for $(v,g)$.
\end{rmk} 

\begin{rmk}\label{rmk: sol1}
    Suppose we instead have a pair $(v,g)$ on the edges of a graph, we can create a path $f$ on the vertices as follows:
    \begin{equation}\label{eq: path}
        f(t)_x = \int_0^t \sum_{k \in \EE} \omega_{x,k}v(\tau)_kg(\tau)_k d\tau + f(0)_x
    \end{equation}
    and for shorthand we denote $f = \int \Omega \cdot(vg)$. By construction, $(f,v,g)$ will satisfy \eqref{eq: original}. Conversely, for any path $f(t)$ on the vertices, no matter which $(v,g)$ pair we choose as a solution to \eqref{eq: original} we can reconstruct $f$ by using \eqref{eq: path}.
\end{rmk}
The lack of a unique inverse of $\Omega$ prevents us from obtaining $g$ in terms of $f$ and $v$ or $v$ in terms of $f$ and $g$. However we can use Remark \ref{rmk: sol1}, to limit our search for pairs $(v,g)$ living on the edges of $\GG$. 
\begin{prop}[Reduced Formulation]\label{thm: simple}
    We have that
    \begin{equation*}
        V_q(f(0),f(1))= \inf_{v,g} \left\{ \mathcal{I}_q(v,g) : f(1) - f(0) = \Omega \textstyle \int_0^1 gv dt \right\}
    \end{equation*}
\end{prop}
\begin{proof}
    By Remark \ref{rmk: sol1}, we have a surjection $(v,g) \mapsto f$ such that $\partial_t f = \Omega \cdot (vg)$, so searching over triples $(f,v,g)$ accounts to searching over $(v,g)$ and inducing a path $f$ via Equation \eqref{eq: path}. However we require the constructed $f$ to truly interpolate between $f(0)$ and $f(1)$, which can only happen if 
    \begin{equation}\label{eq: reduced1}
        \int_0^1 \sum_{k \in \EE} \omega_{x,k}g(t)_kv(t)_k dt = f(1) - f(0).
    \end{equation}
    As shorthand we denote $\Omega \cdot \int_0^1 gv dt = f(1) - f(0)$.
\end{proof}
This formulation is connected to the Beckmann formulation \cite{beckmann1952continuous}, albeit with time-dependent parameters.
\begin{rmk}\label{rmk: sol3}
We have a new condition given by \eqref{eq: reduced1}, so we should also investigate the structure of solutions, fortunately we can inherit most of the machinery from Remark \ref{rmk: sol2}, as we are still inverting an incidence matrix. We have that $(v,g)$ must satisfy
\begin{equation}\label{eq: reducedcond}
    \int_0^1 vg dt = P \cdot [\tilde{f}(1) - \tilde{f}(0)] + \epsilon
\end{equation}
for some $P$ such that $\Omega \cdot P = I$ and $\epsilon \in \ker \Omega$.
\end{rmk}
\addtolength{\topskip}{2pt}
\subsection{Benamou-Brenier on the Vertices of a graph}
The multiple possible inversions of $\Omega$ explain why we cannot extend the argument of Section \ref{sect: tree} to a graph. In fact even constant speed solutions defined in Lemma \ref{lem: const} might not minimize $V_q$, never mind integrate to $W_1$. However we can demonstrate that constant speed solutions are still minimal.
\begin{lem}[Generalised Constant Speed Solutions]\label{lem: const2}
    Given a graph $\GG$, we say a pair $(v,g)$ which satisfies \eqref{eq: reduced1} is a \textit{Constant Speed Solution} if $|v_k| = |v|$ for some function $|v|: [0,1] \rightarrow \mathbb{R}$. For any pair $(v,g)$ we have a constant speed solution $(v^*,g^*)$ such that $\mathcal{I}_q(v,g) \geq \mathcal{I}_q(v^*,g^*)$. Additionally, we can find another constant speed solution $(v^{**},g^{**})$ such that $|v^{**}_k| = |v^{**}| = c$ for $c \in \mathbb{R}$ and $\mathcal{I}_q(v^{*},g^{*}) \geq \mathcal{I}_q(v^{**},g^{**})$.
\end{lem}
\begin{proof}
     See Appendix \ref{appendix: lem: const2}.
\end{proof}
The key takeaway from Lemma \ref{lem: const2} is that if we find any pair $(v,g)$ which achieves $W_1$, we can transform it via $(v^*,g^*)$ and $(v^{**},g^{**})$ to $W_1$ achieving pairs which have desirable properties. One of these properties will be key in the main result of Section \ref{sec: 4}. On a tree, we have one inverse (up to choice of root node) for Remark \ref{rmk: sol3} and we know the constant speed solutions are explicitly defined via Lemma \ref{lem: const}. 
\newline \newline 
We now present the main result of the paper: showing that the Benamou-Brenier formula holds for $q \geq 1$ on a general connected graph.
\begin{thm}[Benamou-Brenier across a Graph]\label{thm: benamou}
    For a connected graph $\GG$, for any $q \geq 1$, we can express the Wasserstein-$1$ distance between $f(0)$ and $f(1)$ by
    \begin{equation}
        \inf_{f,v,g} \left\{\left( \int_0^1 \sum_{k \in \EE} g(t)_k |v(t)_k|^q dt\right)^{\frac{1}{q}} : \partial_t f = \Omega \cdot (vg) \right\}
    \end{equation}
    where $\Omega$ is an arbitrary incidence matrix of $\GG$, and the infimum is achieved for at least one triple $(f,v,g)$.
\end{thm}
\begin{proof}
    See Appendix \ref{appendix: thm: benamou}.
\end{proof}
\begin{rmk}[Wasserstein Tree Distance]
    The proof of Theorem \ref{thm: benamou} is independent of the proof for the tree, meaning that the $W_1$ distance on a tree can be derived from Theorems \ref{thm: benamoutree} and \ref{thm: benamou}.
\end{rmk}
\section{Characterisation of Geodesics}\label{sec: 4}
On $\mathbb{R}^d$, we can achieve the minima of \eqref{eq: benamou} by a pair $(f,v)$ including a constant speed geodesic $f$. We also know that constant speed $W_p$-geodesics between $f(0)$ and $f(1)$ are the McCann displacement interpolations -- which are induced by an optimal couplings between $f(0)$ and $f(1)$.
\newline\newline
In Examples \ref{examp: bin}, \ref{examp: poi} and \ref{examp: star graph}, we saw that the minimizing triple $(f,v,g)$ for \eqref{eq: I} included a constant speed geodesic $f$ in much the same way. We also know in general that the path $f$ can be obtained from $(v,g)$ via \eqref{eq: path} much like the McCann interpolations. We now establish our own characterisation of these geodesics.

\begin{prop}[Geodesics]\label{thm: geod}
Given a pair $(v,g)$ which achieves $W_1$ in Proposition \ref{thm: simple}, and satisfies $|v_k| = |v| = c$, the induced path $f$ is a constant speed $W_1$ geodesic. Conversely, if $f$ is a constant speed geodesic, then there exists $(v,g)$ such that  $|v_k| = |v| = c$ which achieves $W_1$ and induces $f$.
\end{prop}
\begin{proof}
    Let $(v,g)$ be a constant speed solution that minimises $I_1$ and $|v_k| = |v| = c$, then induce the path $f(t)$. For an arbitrary subinterval $[s,t]$, the pair $(v,g)$ restricted to $[s,t]$ is still a solution to $\eqref{eq: continuity}$. It can also be shown by contradiction that $(v,g)$ restricted to $[s,t]$ also minimises $I_1$ across $[s,t]$, (see Appendix \ref{appendix: sol}). Since $(v,g)$ is constant speed, we have
    \begin{align*}
        W_1(f(s),f(t)) = \int_s^t |v| d\tau
        &= (t-s)\int_0^1 |v| d\tau \\ &= (t-s)W_1(f(0),f(1))
    \end{align*}
    and so $f$ is a constant speed geodesic. Conversely, suppose $f$ is a constant speed geodesic, so $|\dot{f}|_{W_1} = W_1(f(0),f(1))$. We then use the remarkable result that (if we let $W^{\mathcal{H}}_1$ denote the Wasserstein distance across the graph $\mathcal{H}$) there exists some spanning tree $\TT \subset \GG$ such that $W^{\TT}_1(f(0),f(1)) = W^{\GG}_1(f(0),f(1))$, (see \cite{guy2019kantorovich}). We can choose $(v,g)$ as the constant speed solution on $\EE_{\TT}$ (as in Lemma \ref{lem: const}), and $v = g = 0$ for any $\EE \backslash \EE_{\TT}$. This pair achieves $W_1^{\TT} = W_1^{\GG}$ and since $(f,v,g)$ satisfy the transport equation, $f$ is induced by $(v,g)$.
\end{proof}
\begin{colry}[Convex Interpolations are $W_1$ Geodesics]\label{thm: conv}
    Let $f(0),f(1)$ be two distributions on the vertices of a graph $\GG = (\VV,\EE)$, then $f(t) = (1-t)f(0) + tf(1)$ is a constant speed geodesic.
\end{colry}
\begin{proof}
    Let $J$ be a minimising flow in the Beckmann formulation, then we choose, $v_kg_k = J_k$.  This achieves $W_1$ and then $f(t) = \int_0^t \Omega \cdot J d\tau = t(f(1) - f(0)) + f(0)$ is a constant speed $W_1$ geodesic.
\end{proof}
\section{Conclusion}
We have extended the work of Hillion and Johnson (see \cite[Lemma 3.6]{hillion2016}) to demonstrate that the discrete Benamou-Brenier formulation of $W_1$ holds in generality on $\mathbb{Z}$, on trees, and eventually on graphs. The tree case lends itself particularly well, with an exact form for the minimizing triples in terms of tail distributions of constant speed , (see Lemma \ref{thm: tail} and Theorem \ref{thm: benamoutree}).  \newline \newline 
While the minimal Beckmann flow provides us with a convex interpolation (see Corollary \ref{thm: conv}), we see that on $\mathbb{Z}$ and even cycles we have minimizing triples $(f,v,g)$ with more exotic geodesics (see Examples \ref{examp: bin}, \ref{examp: poi} and Example 6 in  Appendix H). In fact the existence of Corollary \ref{thm: conv} and Example \ref{examp: bin} tell us that
\begin{align*}
    f_1(t) &= \bin(n,(1-t)p + tq)\\
    f_2(t) &= (1-t)\bin(n,p) + t\bin(n,q)
\end{align*}
are two valid constant speed geodesics between $\bin(n,p)$ and $\bin(n,q)$. We formalise this characterisation of $W_1$ geodesics in Proposition \ref{thm: geod}.
\newpage
\appendix
\subsection{Proof of Lemma \ref{lem: const} for $q \geq 1$}\label{appendix: lem: const}
\begin{proof}
    Let us compute $|\dot{f}|_{W_1}$ directly, we have that 
    \begin{align*}
        |\dot{f}|_{W_1} &= \lim_{h \rightarrow 0} \frac{W_1(f(t+h),f(t))}{h} \\
        &=\lim_{h \rightarrow 0} \frac{1}{h}\sum_{x \in \VV} |F(t+h)_x - F(t)_x| \\
        &=  \sum_{x \in \VV} \left|\lim_{h \rightarrow 0}\frac{F(t+h)_x - F(t)_x}{h}\right| \\
        &= \sum_{x \in \VV}|\partial_t F(t)_x| =  |v|(t).
    \end{align*}
    Now assume that $f$ is a constant speed geodesic, we know from Definition \ref{defn: constant} that  \begin{equation*}
        W_1(f(t+h),f(t)) = |h| W_1(f(0),f(1)),
    \end{equation*}
     re-arranging and taking limits we have 
    \begin{equation*}
        |v| = |\dot{f}|_{W_1} = \lim_{h \rightarrow 0} \frac{W_1(f(t+h),f(t))}{|h|} = W_1(f(0),f(1)).
    \end{equation*}
    So overall we have
    \begin{align*}
        \mathcal{I}_q(v,g)^q &= \int_0^1 \sum_{k \in \EE}g(t)_kv(t)_k^q dt\\
        &= \int_0^1 \sum_{k \in \EE}g(t)_k|v|^q dt \\
        &= \int_0^1 |v|^q dt \\
        &= |v|^q = W_1(f(0),f(1))^q
    \end{align*}
    and so $\mathcal{I}_q(v,g) = W_1(f(0),f(1))$.
\end{proof}
\subsection{Proof that Convex Interpolation is a Geodesic on a Tree}\label{appendix: geodesic}
\begin{proof}
    Consider a convex interpolation $f(t) = (1-t)f(0) + tf(1)$, then the tail distribution is similarly convex $F(t) = (1-t)F(0) + tF(1)$. Then we directly compute:
    \begin{align*}
        W_1(f(s),f(t)) &= \sum_{x \in \VV}|F(s)_x - F(t)_x| \\
        &=  \sum_{x \in \VV}|(1-s)F(0)_x + sF(1)_x \\
        &\quad \quad \quad - (1-t)F(0)_x - tF(1)_x| \\
         &= \sum_{x \in \VV}|(t-s)(F(1)_x - F(0)_x)|\\
         &= |t-s|  \sum_{x \in \VV}|F(1)_x - F(0)_x| \\
         &= |s-t|W_1(f(0),f(1)).
    \end{align*}
    So $f$ is a constant speed geodesic.
\end{proof}
\subsection{Proof of Theorem  \ref{thm: benamoutree} for $q \geq 1$}\label{appendix: thm: benamoutree}
\begin{proof}
    An immediate consequence of \eqref{eq: tree} and \eqref{eq:tailpde} is we are given an integral formulation of $W_1$. Remarkably, if $(f,v,g)$ is \textit{any} triple which satisfies \eqref{eq: original}, then we can write
    \begin{align*}
        W_1(f(0),f(1) &= \sum_{x \in \VV} | F(1)_x - F(0)_x| \\
        &= \sum_{x \in \VV} \left| \int_0^1 \partial_t F(t)_x dt \right| \\
        &= \sum_{k \in \EE} \left| \int_0^1 g(t)_kv(t)_k dt \right|
    \end{align*}
    which transforms a sum over tails to a sum over edges. Although this goes further, by applying both Jensen's and the Cauchy Schwarz inequalities, we can bound $W_1$ above by $V_2$.
    \begin{align*}
        W_1(f(0),f(1))^q &= \left( \sum_{k \in \EE} \left| \int_0^1 g(t)_kv(t)_kdt \right|\right)^q \\
        &\leq \left( \int_0^1 \sum_{k \in \EE} |g(t)_kv(t)_k| dt \right)^q\\
        &\leq \int_0^1 \left( \sum_{k \in \EE} |g(t)_kv(t)_k|\right)^q dt \\
        &\leq \int_0^1 \sum_{k \in \EE} g(t)_k|v(t)_k|^q dt  = \mathcal{I}_q(v,g)^q
    \end{align*}
    So we have that $V_q \geq W_1$ too. However, Lemma \ref{lem: const}, says that a constant speed solution achieves $W_1$, so we in fact have that $V_q(f(0),f(1)) = W_1(f(0),f(1))$ and we are done.
\end{proof}
\subsection{Proof of Lemma \ref{lem: const2}}\label{appendix: lem: const2}
\begin{proof}
    Let $(v,g)$ be a solution to \eqref{eq: reduced1}, we induce a path $f$, and then since $(f,v,g)$ satisfy the transport equation, for some $(P,\epsilon)$ we have $vg = P \cdot \partial_t f + \epsilon$. Then define $(v^*,g^*)$ by 
    \begin{align*}
        v^*_k &= \text{sign}(P\cdot \partial_t f + \epsilon)_k|v^*| \\
        g^*_k &= |P \cdot \partial_t f + \epsilon|_k/ |v^*| \\
        |v^*| &= \sum_{k \in \EE} |P \cdot \partial_t f + \epsilon|_k
    \end{align*}
    then we know $\int_0^1 v^*g^* dt = P \cdot (f(1) - f(0)) + \epsilon$ and so $\int_0^1 \Omega v^*g^* dt = f(1) - f(0)$. Finally we have
    \begin{align*}
        \mathcal{I}_q(v,g)^q &= \int_0^1 \sum_{k \in \EE} g_k|v_k|^q dt \\
        &\geq \int_0^1 \left( \sum_{k \in \EE} |g_kv_k|\right)^q dt \\
        &= \int_0^1 \left( \sum_k |P \cdot \partial_t f + \epsilon|_k \right)^q dt \\
        &= \int_0^1 |v^*|^q dt \\
        &= \int_0^1 \sum_k g^*_k |v_k^*|^q dt = \mathcal{I}_q(v^*,g^*)^q.
    \end{align*}
    Now consider an increasing function $\phi: [0,1] \rightarrow [0,1]$ such that $\phi(0) = 0$ and $\phi(1) = 1$ and the following pair:
    \begin{align*}
        v^{**}(t)_k =\phi'(t) \cdot v^{*}(\phi(t))_k, \quad g^{**}(t)_k = g^{*}(\phi(t))_k.
    \end{align*}
    By integration by substitution we know 
    \begin{equation*}
        \int_0^1 \Omega \cdot v^{**}g^{**}dt = \int_0^1 \Omega\ \cdot v^*g^* dt = f(1) - f(0).
    \end{equation*}
    So $(v^{**},g^{**})$ is a valid solution. We consider when 
    \begin{equation}\label{eq: diff}
        \left|\phi'(t) \frac{d v^*}{d\phi} \right| =\phi'(t) \left| \frac{dv^*}{d\phi}\right|= c.
    \end{equation}
    We denote $s(\phi) = \left| \frac{dv^*}{d\phi}\right|$ and $S(y) = \int_0^y s(y) dy$ then solving \eqref{eq: diff} we have $S(\phi(t)) = ct+d$. We know that  $c  = S(\phi(1)) - S(\phi(0))= S(1) - S(0)$ and $d = S(\phi(0))= S(0) $. So the choice of 
    \begin{equation*}
        \phi(t) \coloneqq S^{-1}[(S(1)-S(0))t + S(0)]
    \end{equation*}
    gives us $|v^{**}| = c = S(1) - S(0) = \int_0^1 s(t) dt$. Note that $S^{-1}$ exists because $s(y) > 0$ so $S$ is increasing. This parametrisation has time-independent speed and by the chain rule and Jensen's inequality we have 
    \begin{align*}
        \mathcal{I}_q(v^{**},g^{**})^q &=\int_0^1 |v^{**}|^q dt \\
        &= |v^{**}|^q \\
        &= \left( \int_0^1 |v^{**}| dt \right)^{q}\\
        &= \left( \int_0^1 \phi'(t) |v^*(\phi(t))| dt \right)^q \\
        &=\left(\int_0^1 |v^{*}(s)| ds\right)^q \\
        &\leq \int_0^1 |v^{*}(s)|^q ds = \mathcal{I}_q(v^*,g^*)^q
    \end{align*}
    
\end{proof}
\subsection{Proof of Theorem \ref{thm: benamou}}\label{appendix: thm: benamou}
\begin{proof}
    We first use the reduced form for $V_q$:
    \begin{equation*}
        V_q(f(0),f(1))= \inf_{v,g} \left\{ \mathcal{I}_q(v,g) : f(1) - f(0) = \Omega \textstyle \int_0^1 gv dt \right\}
    \end{equation*}
    We can express the Beckmann formulation of the $W_1$ distance as follows:
    \begin{equation*}
        W_1(f(0),f(1)) = \inf_J \left\{\sum_{k \in \EE} |J_k| : \Omega \cdot J = f(1) - f(0) \right \}
    \end{equation*}
    We begin by showing that $V_q \geq W_1$, let $(v,g)$ satisfy
    \begin{equation*}
        \Omega \cdot \int_0^1 gv dt = f(1) - f(0).
    \end{equation*}
    We see that identifying $J \coloneqq \int_0^1 gv dt$ gives a valid solution to Beckmann's formulation. Then we have that 
    \begin{align*}
        \mathcal{I}_q(v,g)^q &= \int_0^1 \sum_{k} g(t)_k|v(t)_k|^q dt \\
        &\geq \int_0^1 \left( \sum_k |g(t)_kv(t)_k| \right)^q dt \\
        &\geq \left( \int_0^1 \sum_k |g(t)_kv(t)_k| dt\right)^q\\
        &\geq \left( \sum_k \int_0^1 |g(t)_kv(t)_k| dt \right)^q \\
        &\geq \left(\sum_k \left|\int_0^1 g(t)_kv(t)_k  dt\right| \right)^q\\
        &= \left( \sum_k |J_k|\right)^q \\
        &\geq W_1(f(0),f(1))^q
    \end{align*}
    So we have that $V_q \geq W_1$. Conversely, choose $J$ a minimising solution (achieves $W_1$) to the Beckmann formulation, then we can choose 
    $|v| = \sum_{k} |J_k|$, $v(t)_k = \text{sign}(J_k)|v|$ and $g(t)_k = |J_k|/|v|$, which satisfies $g(t)_k v(t)_k = J_k$. Additionally this pair $(v,g)$ is constant in time, so we have that 
    \begin{equation*}
        \Omega \cdot \int_0^1 gv dt = \Omega \cdot (gv) = \Omega \cdot J = f(1) - f(0).
    \end{equation*}
    So $(v,g)$ satisfy the condition for $V_q$, then we have that 
    \begin{align*}
        W_1(f(0),f(1))^q &= \left(\sum_{k}|J_k|\right)^q \\
        &= |v|^q \\
        &= \sum_k g(t)_k|v|^q \\
        &= \sum_k g(t)_k |v(t)_k|^q \\
        &= \int_0^1 \sum_k g(t)_k |v(t)_k|^q dt \\
        &\geq V_q(f(0),f(1))^q
    \end{align*}
    So we additionally have $W_1 \geq V_q$, so overall we have that for any $f(0)$ and $f(1)$, we have 
    \begin{equation*}
        W_1(f(0),f(1)) = V_q(f(0),f(1)).
    \end{equation*}
\end{proof}
\subsection{Proof of Local Minimisation}\label{appendix: sol}
\begin{proof}
    Suppose $(v,g)$ is a minimising pair for $I_1$ across $[0,1]$, i.e we have that 
    \begin{align*}
        \mathcal{I}_1(v,g) = \int_0^1 \sum_{k \in \EE}v_kg_k dt &= \min_{(v_0,g_0)} \mathcal{I}_1(v_0,g_0) \\
        &= W_1(f(0),f(1)).
    \end{align*}
    Then we denote
    \begin{equation*}
        \mathcal{I}_q^{[s,t]}(v,g) = \left(\int_s^t \sum_{k \in \EE} g(t)_k |v(t)|_k^q dt \right)^{\frac{1}{q}}.
    \end{equation*}
    This implies that $\mathcal{I}_1^{[0,1]} = \mathcal{I}_1$. Suppose $(v,g)$ is not minimal across some $[s,t]$ for $\mathcal{I}_1$, then there exists some $(v^*,g^*)$ such that 
    \begin{equation*}
        \mathcal{I}_1^{[s,t]}(v^*,g^*) < \mathcal{I}_1^{[s,t]}(v,g).
    \end{equation*}
    From here we construct a new solution across $[0,1]$:
    \begin{equation*}
        (v^{**}(\tau)_k,g^{**}(\tau)_k) = \begin{cases}
            (v^*(\tau)_k,g^*(\tau)_k) &: \tau \in [s,t] \\
            (v(\tau)_k,g(\tau)_k) &: \tau \not \in [s,t]
        \end{cases}
    \end{equation*}
    Then by decomposing
    \begin{align*}
        \mathcal{I}_1^{[0,1]} = \mathcal{I}_1^{[0,s]} + \mathcal{I}_1^{[s,t]} + \mathcal{I}_1^{[t,1]} 
    \end{align*}
    we have that 
    \begin{align*}
        \mathcal{I}_1^{[0,1]}(v^{**},g^{**}) &< \mathcal{I}_1^{[0,s]}(v,g) + \mathcal{I}_1^{[s,t]}(v,g) + \mathcal{I}^{[t,1]}(v,g) \\
        &= \mathcal{I}_1^{[0,1]}(v,g)
    \end{align*}
    which contradicts the minimality of $(v,g)$.
\end{proof}
\subsection{Supplementary Results}\label{appendix: supp}
The following results are true by virtue of $V_q = W_1$, although contain useful insights nevertheless, including a proof that Wasserstein-$1$ is bounded below by the total variation.
\begin{lem} $V_q(f(0),f(1)) \geq \text{TV}(f(0),f(1))$ for all $q \geq 1$, where $\text{TV}$ is the total variational distance.
\end{lem}
\begin{proof}
   Consider the reduced form of $V_q$, and let $(v,g)$ satisfy
   \begin{equation*}
       \Omega \cdot \int_0^1 vg dt = f(1) - f(0).
   \end{equation*}
   Writing this out explicitly we have
   \begin{align*}
       |f(1)_x - f(0)_x| &= \left|\int_0^1 \sum_{k \rightarrow x} g(t)_kv(t)_k - \sum_{x \rightarrow k}g(t)_kv(t)_k dt \right| \\
       &\leq \int_0^1 \left|\sum_{k \rightarrow x}g(t)_kv(t)_k - \sum_{ x \rightarrow k}g(t)_k v(t)_k \right| dt \\
       &\leq \int_0^1 \sum_{k \rightarrow x}|g(t)_kv(t)_k| + \sum_{x \rightarrow k} |g(t)_kv(t)_k| dt.
   \end{align*}
   But considering the identity
   \begin{align*}
       \sum_x \sum_{k:k \rightarrow x}h_k + \sum_x \sum_{k:x \rightarrow k} h_k = 2\sum_k h_k
   \end{align*}
   we have that 
   \begin{align*}
       \sum_x |f(1)_x - f(0)_x| &\leq 2\int_0^1 \sum_k |g(t)_k v(t)_k| dt.
   \end{align*}
   Finally we have that 
   \begin{align*}
       \mathcal{I}_q(v,g)^q &\geq \int_0^1 \sum_k g(t)_k |v(t)_k|^q dt \\
       &\geq \int_0^1 \left(\sum_k |g(t)_kv(t)_k| \right)^q dt\\
       &\geq \left( \int_0^1 \sum_k |g(t)_kv(t)_k| \right)^q \\
       &\geq \left(\frac{1}{2}\sum_x |f(1)_x - f(0)_x| \right)^q = \text{TV}(f(0),f(1))^q.
   \end{align*}
   So $V_q \geq \text{TV}$ too. 
\end{proof}
\begin{prop}$V_q$ defines a metric on $\mathcal{P}(\VV)$.
\end{prop}
\begin{proof}
    We begin with positive definiteness, we know $\mathcal{I}_q(v,g) \geq 0$ for all $(v,g)$, so $V_q \geq 0$. Additionally if $V_q(f(0),f(1)) = 0$, then $\text{TV}(f(0),f(1)) = 0$ too, so $f(0) = f(1)$ since the total variation is a metric. Now we prove symmetry. For any pair $(v,g)$ which satisfies
    \begin{equation*}
        \Omega \cdot \int_0^1 gv dt = f(1) - f(0)
    \end{equation*}
    we can build a pair $(\tilde{g},\tilde{v})$ by $\tilde{g}(t) = g(1-t)$ and $\tilde{v}(t) = -v(1-t)$ which satisfies \begin{align*}
        \Omega \cdot \int_0^1 \tilde{g}\tilde{v} dt &= -\Omega \cdot \int_0^1 g(1-t)v(1-t) dt \\
        &= -\Omega \cdot \int_0^1 g(t)v(t) dt \\
        &= f(0) - f(1).
    \end{align*}
    So $(\tilde{g},\tilde{v})$ is valid for $V_q(f(1),f(0))$. Similarly we can see that $\mathcal{I}_q(\tilde{v},\tilde{g}) = \mathcal{I}_q(v,g)$ too. If we consider a sequence $(g^{(n)},v^{(n)})_n$ such that $\lim_{n \rightarrow \infty} \mathcal{I}_q(v^{(n)},g^{(n)}) = V_q(f(0),f(1))$, then we induce a sequence $(\tilde{v}^{(n)},\tilde{g}^{(n)})$ too, then
    \begin{align*}
        V_q(f(0),f(1)) &= \lim_{n \rightarrow \infty}\mathcal{I}_q(v^{(n)},g^{(n)})\\
        &=\lim_{n \rightarrow \infty} \mathcal{I}_q(\tilde{v}^{(n)},\tilde{g}^{(n)}) \\
        &\geq V_q(f(1),f(0)).
    \end{align*}
    We can perform this process in reverse too, so $V_q(f(0),f(1)) = V_q(f(1),f(0))$. Finally we tackle the triangle inequality. We consider $f(0),f^*$ and $f(1)$, then let $(v,g)$ be a pair between $f(0)$ and $f^*$ and $(\tilde{v},\tilde{g})$ be a pair between $f^*$ and $f(1)$, then for some $\rho \in [0,1]$ we define a piecewise solution
    \begin{equation*}
    \hat{v}(t),\hat{g}(t) = 
        \begin{cases}
            \frac{1}{\rho}v\left(\frac{t}{\rho}\right),g\left(\frac{t}{\rho}\right) &: t \in [0,\rho] \\
            \frac{1}{1-\rho}\tilde{v}\left(\frac{t- \rho}{1- \rho} \right),\tilde{g}\left(\frac{t-\rho}{1-\rho}\right) &: t \in [\rho,1]
        \end{cases}
    \end{equation*}
    It can be shown that 
    \begin{equation*}
        \mathcal{I}_q(\hat{v},\hat{g})^q = \frac{1}{\rho}\mathcal{I}_q(v,g)^q + \frac{1}{1 - \rho}\mathcal{I}_q(\tilde{v},\tilde{g})^q
    \end{equation*}
    and we choose $\rho$ such that 
    \begin{equation*}
        \mathcal{I}_q(\hat{v},\hat{g})^q = \left(\mathcal{I}_q(v,g) + \mathcal{I}_q(\tilde{v},\tilde{g}) \right)^q
    \end{equation*}
    and hence 
    \begin{equation*}
         \mathcal{I}_q(\hat{v},\hat{g}) = \mathcal{I}_q(v,g) + \mathcal{I}_q(\tilde{v},\tilde{g}).
    \end{equation*}
    Then if we let $(g^{(n)},v^{(n)})_n$ and $(\tilde{g}^{(n)},\tilde{v}^{(n)})_n$ be sequences such that $\mathcal{I}_q(g^{(n)},v^{(n)}) \rightarrow V_q(f(0),f^*)$ and $\mathcal{I}_q(\tilde{g}^{(n)},\tilde{v}^{(n)}) \rightarrow V_q(f^*,f(1))$ respectively, then 
    \begin{align*}
        V_q(f(0),f(1)) &\leq \lim_{n \rightarrow \infty} \mathcal{I}(\hat{v}^{(n)},\hat{g}^{(n)}) \\
        &= V_q(f(0),f^*) + V_q(f^*,f(1))
    \end{align*}
    and we have the triangle inequality. Finally we motivate why we can choose $\rho$ in such a way. We are essentially choosing $\rho$ such that 
    \begin{equation*}
        \frac{1}{\rho}a^q + \frac{1}{1-\rho}b^q = (a+b)^q.
    \end{equation*}
    This is a quadratic in $\rho$ which has valid solutions in $[0,1]$.
\end{proof}
The following is a worked example for finding valid minimizing solutions on a $4$-cycle.
\begin{exam}\label{examp: square}
 Consider the following distribution $f$ on a square generated by two probabilities $p$ and $q$, we consider the Wasserstein distance between when $p = p_0,q = q_0$ and $p = p_1,q = q_1$ respectively.
    
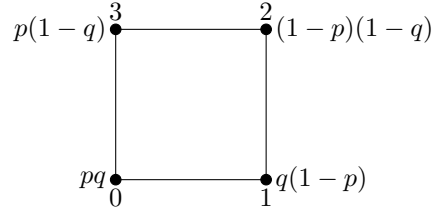
\begin{center}
        \centering
        \begin{tikzpicture}
  \def\spacing{2}  

  \coordinate (P0) at (0, 0);
  \coordinate (P1) at (\spacing, 0);
  \coordinate (P2) at (\spacing, \spacing);
  \coordinate (P3) at (0, \spacing);

  \foreach \i in {0,1,2,3} {
    \filldraw[black] (P\i) circle (2pt);
  }
  \node[left] at (P0) {$pq$};
  \node[right] at (P1) {$q(1-p)$};
  \node[right] at (P2) {$(1-p)(1-q)$};
  \node[left] at (P3) {$p(1-q)$};
  \node[below] at (P0) {$0$};
  \node[below] at (P1) {$1$};
  \node[above] at (P2) {$2$};
  \node[above] at (P3) {$3$};
  \draw (P0) -- (P1) -- (P2) -- (P3) -- cycle;
\end{tikzpicture}
    \end{center}
Since these are product measures, we can show that $W_1(f(0),f(1)) = |p_1 - p_0| + |q_1 - q_0|$. We can express the transport equation as:
\begin{align*}
    \partial_t f(t)_0 &= p'q + pq' &&=  -g_{01}v_{01} - g_{03}v_{03} \\
    \partial_t f(t)_1 &= q'(1-p) - p'q &&= \quad g_{01}v_{01} - g_{12}v_{12}\\
    \partial_t f(t)_2 &= -p'(1-q) - q'(1-p) &&= \quad  g_{12}v_{12} + g_{32}v_{32} \\
    \partial_t f(t)_3 &= p'(1-q) - pq' &&=\quad g_{03}v_{03} - g_{32}v_{32}
\end{align*}
Since we have one cycle, we know $\text{Nullity}(\Omega) =1$, so we have one degree of freedom in solutions which are written as
\begin{align*}
    g_{01}v_{01} &=-p'q + \epsilon \\
    g_{03}v_{03}&= -pq' - \epsilon\\
    g_{12}v_{12} &= q'(1-p) - \epsilon\\
    g_{32}v_{32}&= p'(1-q) + \epsilon 
\end{align*}
We choose a specific constant speed solution $(v,g)$ by assigning $|v| = |p'| + |q'|$, then defining the velocity $v$ by:
\begin{align*}
    v_{01} = v_{32} = \text{sign}(p')|v|, \quad v_{03} = v_{12} = \text{sign}(q')|v|
\end{align*}
and defining the edge distribution $g$ by
\begin{align*}
 |v|g_{01} &= |p'|q, &\quad |v|g_{32} &= |p'|(1-q)\\
    |v|g_{03} &= |q'|p, &\quad |v|g_{12} &= |q'|(1-p) \\
\end{align*}
When $p'$ and $q'$ are constant, $|v|$ is additionally constant in time -- as well as edge-wise. This constant speed solution achieves $W_1$, so by Proposition \ref{thm: geod}, the induced path $f(t)$ is a geodesic. This corresponds to interpolating $p(t) = (1-t)p_0 + tp_1$ and $q(t) = (1-t)q_0 + tq_1$.
\end{exam}
\newpage 
\bibliography{refs}
\end{document}